\documentclass{article}
\usepackage[T1]{fontenc}    %

\usepackage{amsmath,amssymb,amsthm}
\usepackage{bbm} %

\usepackage{fullpage}
\usepackage{xcolor}

\usepackage{tikz}
\usepackage{mdframed}
\usepackage{thmtools} %

\makeatletter
\@ifundefined{newcounteralias}{}{%
  \renewcommand\thmt@autorefsetup{\@xa\def\csname\thmt@envname autorefname\@xa\endcsname\@xa{\thmt@thmname}}%
}
\makeatother

\usepackage[pagebackref,colorlinks=true,urlcolor=blue,citecolor=blue,
    linktocpage=true,linkcolor=purple,pdfstartview=FitH]{hyperref}
\usepackage[nameinlink]{cleveref}

\newtheorem{theorem}{Theorem}[section]

\newtheorem{lemma}[theorem]{Lemma}
\newtheorem{claim}[theorem]{Claim}

\newtheorem{definition}[theorem]{Definition}

\crefname{theorem}{theorem}{theorems}
\Crefname{theorem}{Theorem}{Theorems}
\crefname{conjecture}{conjecture}{conjectures}
\Crefname{conjecture}{Conjecture}{Conjectures}
\crefname{lemma}{lemma}{lemmas}
\Crefname{lemma}{Lemma}{Lemmas}
\crefname{claim}{claim}{claims}
\Crefname{claim}{Claim}{Claims}
\crefname{corollary}{corollary}{corollaries}
\Crefname{corollary}{Corollary}{Corollaries}
\crefname{definition}{definition}{definitions}
\Crefname{definition}{Definition}{Definitions}
\crefname{proposition}{proposition}{propositions}
\Crefname{proposition}{Proposition}{Propositions}
\crefname{fact}{fact}{facts}
\Crefname{fact}{Fact}{Facts}
\crefname{example}{example}{examples}
\Crefname{example}{Example}{Examples}
\crefname{assumption}{assumption}{assumptions}
\Crefname{assumption}{Assumption}{Assumptions}
\crefname{observation}{observation}{observations}
\Crefname{observation}{Observation}{Observations}
\crefname{remark}{remark}{remarks}
\Crefname{remark}{Remark}{Remarks}
\crefname{question}{question}{questions}
\Crefname{question}{Question}{Questions}

\newcommand{\ignore}[1]{}

\newcommand{\Z}{{\mathbb Z}}

\newcommand{\calP}{{\cal P}}
\newcommand{\calQ}{{\cal Q}}

\newcommand{\mbone}{\mathbbm{1}}

\newcommand{\set}[1]{\left\{ #1 \right\}}
\newcommand{\ceil}[1]{\lceil#1\rceil}
\newcommand{\Exp}{\EX}

\newcommand{\given}{\textrm{ s.t. }}

\renewcommand{\leq}{\leqslant}
\renewcommand{\geq}{\geqslant}

\newcommand{\EX}{\mathbb{E}}

\newcommand{\ind}[1]{\mbone\left[ #1 \right]}

\newcommand{\Sec}[1]{\hyperref[sec:#1]{\Cref*{sec:#1}}} %
\newcommand{\Eqn}[1]{\hyperref[eq:#1]{(\ref*{eq:#1})}} %
\newcommand{\Fig}[1]{\hyperref[fig:#1]{Fig.\,\ref*{fig:#1}}} %
\newcommand{\Tab}[1]{\hyperref[tab:#1]{Tab.\,\ref*{tab:#1}}} %
\newcommand{\Thm}[1]{\hyperref[thm:#1]{Theorem\,\ref*{thm:#1}}} %
\newcommand{\Fact}[1]{\hyperref[fact:#1]{Fact\,\ref*{fact:#1}}} %
\newcommand{\Lem}[1]{\hyperref[lem:#1]{Lemma\,\ref*{lem:#1}}} %
\newcommand{\Prop}[1]{\hyperref[prop:#1]{Prop.~\ref*{prop:#1}}} %
\newcommand{\Cor}[1]{\hyperref[cor:#1]{Corollary~\ref*{cor:#1}}} %
\newcommand{\Conj}[1]{\hyperref[conj:#1]{Conjecture~\ref*{conj:#1}}} %
\newcommand{\Def}[1]{\hyperref[def:#1]{Definition~\ref*{def:#1}}} %
\newcommand{\Alg}[1]{\hyperref[alg:#1]{Alg.~\ref*{alg:#1}}} %
\newcommand{\Obs}[1]{\hyperref[obs:#1]{Obs.~\ref*{obs:#1}}} %
\newcommand{\Ex}[1]{\hyperref[ex:#1]{Ex.~\ref*{ex:#1}}} %
\newcommand{\Clm}[1]{\hyperref[clm:#1]{Claim~\ref*{clm:#1}}} %
\newcommand{\Step}[1]{\hyperref[step:#1]{Step~\ref*{step:#1}}} %

\newcommand{\tO}[1]{\tilde{O}(#1)}

\newcommand{\bigO}[1]{O\left(#1\right)}
\newcommand{\bigtO}[1]{\tilde{O}\left(#1\right)}
\newcommand{\bigOm}[1]{\Omega\left(#1\right)}
\newcommand{\bigTh}[1]{\Theta\left(#1\right)}

\newcommand{\TMUL}{\mathsf{T}_{\textsf{MUL}}}

\crefname{algocf}{alg.}{algs.}
\Crefname{algocf}{Alg.}{Algs.}

\crefname{step}{Step}{Steps}
\Crefname{step}{Step}{Steps}

\newcommand{\Comment}[1]{\textcolor{teal}{#1}}

\title{Simple, Faster $(\min, +)$ Matrix Multiplication and APSP}

\author{
Barna Saha\thanks{Supported by NSF HDR TRIPODS Phase II grant 2217058 (EnCORE Institute).}  \\
University of California San Diego  \\
\texttt{bsaha@ucsd.edu}
\and
Christopher Ye\footnotemark[1]  \\
University of California San Diego  \\
\texttt{czye@ucsd.edu}
}

\begin{document}

\pagenumbering{roman}

\maketitle

\thispagestyle{empty}

\begin{abstract} 
    In this note, we provide a deterministic algorithm that computes the $(\min, +)$-Matrix Product $C = A * B$ of two $N \times N$ matrices $A$ and $B$ with polynomially bounded integer entries in $\tilde{O}(N^{2.9964})$ time, improving upon the exponent of $O(N^{2.9995})$ time obtained in the recent breakthrough by Anthropic (Alman and Vassilevska Williams, 2026). 
Our algorithm uses a phased approach first proposed by Bringmann, Grandoni, Saha, and Vassilevska Williams in 2016. 
The phased approach produces an approximation $\tilde{C}$ of $C$ in truly subcubic time such that the number of remaining triples $(i,k,j): A_{i,k}+B_{k,j} < \tilde{C}(i,j)$ is truly subcubic. 
This provided strong evidence that a truly subcubic $(\min, +)$-Matrix Product and therefore, a truly subcubic algorithm for All Pairs Shortest Paths (APSP) may exist. 
However, to find the remaining triples in subcubic time, a sparse certificate is needed whose construction was not known. 
We use the thin matrix product theorem of the recent breakthrough to provide such a certificate.
This gives a simple modular subcubic algorithm for APSP.

\end{abstract}

\clearpage
\pagenumbering{arabic}

\section{Introduction}

All-Pairs Shortest Paths (APSP), which asks given an $N$-vertex graph to compute the distance $\delta(u, v)$ between all pairs of vertices, is one of the central problems in fine-grained complexity.
Classical algorithms for APSP, such as the Floyd–Warshall algorithm, achieve a $O(N^3)$ running time.
A long series of works (e.g. \cite{fredman1976new,Dob1990,Takaoka92,Takaoka98,Han04,pettie2004apsp,ZwickAPSP04,Chan05,Han06,DBLP:journals/siamcomp/Chan10,HanT12}) improved sub-polynomial factors over the classic $O(N^3)$ running time, culminating in an algorithm by Williams \cite{Williams18} with $N^3 / 2^{\Theta(\sqrt{\log N})}$ running time.
Very recently, a major breakthrough by Anthropic (see Alman and Vassilevska Williams \cite{alman2026truly}) gave the first truly subcubic algorithm for APSP, falsifying one of the central hypotheses in Fine-Grained Complexity, obtaining a deterministic algorithm in $O(N^{2.9995})$ time (see also the simplified exposition due to Karingula and Lovett \cite{karingula2026simplified} that achieves $O(N^{3 - \delta})$ for some $\delta > 0$).
\cite{alman2026truly} also obtain a Las Vegas algorithm running in expected time $O(N^{2.998})$ for real-valued APSP.
In light of this breakthrough, one might naturally expect further improvements to the APSP exponent. 

In this note, we give a simple improvement to the exponent of APSP, via a faster algorithm for computing the $(\min, +)$ Product of two $N \times N$ matrices.
The $(\min, +)$ product of two matrices $A, B$, defined $(A* B)[i, j] = \min_{k} A[i, k] + B[k, j]$, is a useful primitive in algorithm design, and is known to be equivalent to APSP up to polylogarithmic factors \cite{williams2018subcubic}.
Our first result is an improved randomized algorithm.

\begin{theorem}
    \label{thm:faster-min-plus}
    There is a Las Vegas algorithm computing the $(\min, +)$-product of two $N \times N$ matrices $A, B$ with polynomially bounded integer entries in  expected time $\bigtO{N^{2.9963432}}$.
\end{theorem}

Using known reductions, this also gives faster algorithms for many problems equivalent to $(\min, +)$-Matrix Multiplication, including APSP, Negative Triangle, Replacement Paths, Second Shortest Simple Path, Metricity, $(\min, +)$-Verification \cite{williams2018subcubic}, Weighted Tree Edit Distance \cite{bringmann2020tree, nogler2025faster}, and 2D Maximum Subarray \cite{tamaki2000algorithms, backurs2016tight}.

Our algorithm uses a modular approach to tackle the $(\min,+)$ product in subcubic time, first proposed by Bringmann, Grandoni, Saha, and Vassilevska Williams  \cite{bringmann2019truly}, and subsequently refined and modified in  \cite{gu2021faster, chi2022faster, jin2026deterministic}. In \cite{bringmann2019truly}, the very first step of the algorithm is to provide an additive approximation of $C$ by recursively scaling and rounding. This additive approximation is used to perturb the entries of $A$ and $B$ so that many entries of $A$ and $B$ become small. This enables computation of fast $(\min,+)$ product over those small entries via fast matrix multiplication \cite{alon}. Surprisingly, as \cite{bringmann2019truly} shows after the perturbation, only a subcubic number of witnesses, triples which can contribute to $C$, are left: this provided a promising route towards refuting the APSP conjecture. 

In order to complete this framework and obtain truly subcubic APSP, one needs a sparse certificate to find those remaining triplets. 
The works of \cite{bringmann2019truly, gu2021faster, chi2022faster, jin2026deterministic} show such sparse certificate exists for various structured matrices such as bounded difference and monotone matrices.
Subsequent work \cite{gu2021faster, chi2022faster, jin2026deterministic} combine the perturbation and sparse certification steps through modular arithmetic over primes.

We develop on the modular approach of \cite{bringmann2019truly,gu2021faster, chi2022faster,jin2026deterministic} and utilize Theorem 25 of the recent breakthrough \cite{alman2026truly} to provide a sparse certificate in truly subcubic time without any structural assumptions. 
This leads to a slightly better running time, and a significantly simpler, easy to understand algorithm. 
We state Theorem 25 of \cite{alman2026truly} below.

\begin{restatable}[Theorem 25 of \cite{alman2026truly}]{theorem}{ThinMatrixProduct}
    \label{thm:av-min-plus}
    For every $\varepsilon < \varepsilon^* := (\ln 4) / (5 \ln 10) = 0.1204\dots$ and every $\kappa > 0$ there is a $\gamma > 0$ such that the following holds.
    Given as input matrices $X \in \Z^{N \times D}$ and $Y \in \Z^{D \times N}$, where $2 \leq D \leq N^{\varepsilon}$, whose entries are integers of absolute value at most $N^{O(1)}$, as well as a set $W$ of at most $N^2/ D^{\kappa}$ positions of an $N \times N$ matrix, the entries $(XY)[I,J]$, $(I,J) \in W$, can be computed deterministically in time $O(N^2 \log^2 D / D^{\gamma})$.
\end{restatable}
        
To invoke \Cref{thm:av-min-plus}, we divide the matrices $A$ and $B$ into blocks and for each block, design a family of new matrices $X$ and $Y$ satisfying the properties of \Cref{thm:av-min-plus} such that computing relevant entries of $(XY)$ helps us certify whether a given block of $A$ and $B$ is a candidate in which we ought to look for a witness.
We stop whenever the first witness is obtained.
Our second result shows that the same exponent can be achieved deterministically.

\begin{theorem}
    \label{thm:det-faster-min-plus}
    There is a deterministic algorithm computing the $(\min, +)$-product of two $N \times N$ matrices $A, B$ with polynomially bounded integer entries in time $\bigtO{N^{2.9963432}}$.
\end{theorem}

For the sake of simplicity, we present our randomized algorithm in its entirety, and then present the deterministic algorithm separately.
We leave as an interesting future direction to further improve the exponent of $(\min, +)$ Matrix Multiplication, e.g. by improving the running time of thin matrix product, or to extend our improvements to real-valued matrices.

\paragraph{Ongoing Work.}
The main goal of this note was to give a simplified, modular approach to $(\min, +)$ Matrix Multiplication and APSP, and to contextualize the new breakthrough through existing frameworks, thereby facilitating a better intuitive understanding. 
Any improvement in exponents is a byproduct of that. 
After prompting ChatGPT with our algorithm, it has suggested further improvements of the running time, as well as extending exponent improvements to real-valued matrices using known reductions \cite{chan2022hardness}, Lemma 37 of \cite{alman2026truly} and re-optimizing the parameters. 
We are trying to understand if these improvements fit within our simple framework.

\paragraph{Statement on AI Usage.} 
For the randomized algorithm, we used publicly available ChatGPT Pro for verification of our approach, optimization of parameters and for numerical calculations. 
The subsequent ideas for the deterministic algorithm are due to ChatGPT Pro with mild supervision.
We have evaluated all proposed arguments, and refined the proof and exposition of the claimed results. The paper was written by the authors, and we take full responsibility for the final manuscript.

\section{Preliminaries}

Let $[U] = \set{1, \dots, U}$.
$\log$ is base $2$ and $\ln$ is base $e$ unless otherwise specified.
For $p \in [0, 1]$, we use $H(p) = - p \ln p - (1 - p) \ln (1 - p)$ to denote entropy.
For a matrix $A$, let $A[i, k]$ denote the $(i, k)$-th entry of $A$.
For matrices $A, B$, let $AB$ denote the standard matrix product and $A * B$ denote the $(\min, +)$ product.
Let $\TMUL(N_1, N_2, N_3)$ denote the time required to compute the product of a $N_1 \times N_2$ matrix and $N_2 \times N_3$ matrix.
We use $\omega$ to denote the matrix multiplication exponent, i.e. the infimum over all constants such that two $N \times N$ matrices can be multiplied in $N^{\omega + o(1)}$ time.
It is known $\omega \leq 9/4$ \cite{openaimath}, although for our purposes, weaker bounds on $\omega$ will suffice.
Our algorithm makes use of the prime number theorem.

\begin{theorem}[Prime Number Theorem]
    \label{thm:prime-number}
    Let $\pi(x)$ denote the number of primes at most $x$. 
    The prime number theorem states that 
    \begin{equation*}
        \pi(x) = \frac{x}{\ln x} (1 + o(1)) \text{.}
    \end{equation*}
    As a corollary, for any sufficiently large integer $R$, the number of primes in the interval $[R, 2R]$ is 
    \begin{equation*}
        \pi(2R) - \pi(R) = \Theta(R / \log R) \text{.}
    \end{equation*}
\end{theorem}

\section{Faster \texorpdfstring{$(\min, +)$}{} Matrix Multiplication}

We present the proof of \Cref{thm:faster-min-plus}.

\begin{proof}[Proof of \Cref{thm:faster-min-plus}]
    By a standard reduction (e.g. shifting all entries by a large constant), we may assume that every input value is a positive integer.
    We follow the framework of previous algorithms for Bounded Difference and Monotone $(\min, +)$ Matrix Multiplication.

    \paragraph{Step 1: Reduction to Verification via Additive Approximation.}
    First, by a standard reduction (see e.g. Section 4.1 of \cite{jin2026deterministic}), we observe that it suffices to reduce to the following witness finding problem.

    \begin{definition}
        \label{def:min-plus-verification}
        Given $A,B,C \in \Z^{N \times N}$, the $(\min, +)$-Witness Equality problem asks to check for all $i, j$, whether there exists $k \in [N]$ such that $A[i, k] + B[k, j] = C[i, j]$.
    \end{definition}
    
    \begin{lemma}[Tight Reduction to Verification]
        \label{lem:verification-reduction}
        Suppose $(\min, +)$-Witness Equality can be computed in time $T(N,U)$ whenever $A, B$ have entries in $[U]$ and $C$ has entries in $[2U]$.
        Then the $(\min, +)$-Product of two $N \times N$
        integer matrices with entries in $[U]$ can be computed in time
        \begin{equation*}
            \bigO{(T(N, U) + N^2) \log U} \text{.}
        \end{equation*}
    \end{lemma}

    We defer the proof of \Cref{lem:verification-reduction} and focus on providing an algorithm for $(\min, +)$-Witness Equality.
    The idea is essentially that from $\ceil{A/2} * \ceil{B/2}$, we can obtain a $+2$-approximation of $A * B$.
    Then, invoking the oracle for $(\min, +)$-Witness Equality $O(1)$ times, we can check all $3$ possible values of $(A * B)[i, j]$ for each $i, j$.
    We say $(i, j)$ is accurate if there exists $k \in [N]$ such that $A[i, k] + B[k, j] = C[i, j]$.
    By \Cref{lem:verification-reduction}, it suffices to check for whether every pair $(i, j)$ is accurate.

    \paragraph{Step 2: Finding a Sparse Certificate.}
    In this step, our goal is to efficiently prune triples $(i, k, j)$ for which $A[i, k] + B[k, j] \neq C[i, j]$. 
    Fix parameters $1 \leq B, S \leq N$ whose values we will determine later and let $P \gg B \log U$ for some sufficiently large constant.
    Split the inner dimension $[N]$ into $O(N/B)$ contiguous blocks of size $B$.
    A block $K$ is a \emph{block witness} for output pair $(i, j)$ if there exists $k \in K$ such that $A[i, k] + B[k, j] = C[i, j]$.
    Since $K$ are disjoint and cover $[N]$, observe that $(i, j)$ is accurate if and only if there is a block witness $K$.

    To find block witnesses efficiently, we begin by finding a set of candidate witnesses.
    Let $p \in [P, 2P]$ be a prime.
    For a block $K$ and prime $p$, we say $K$ is a \emph{candidate block witness} with respect to $p$ for $(i, j)$ if there exists $k \in K$ satisfying
    \begin{equation*}
        A[i, k] + B[k, j] \equiv C[i, j] \mod p \text{.}
    \end{equation*}
    We have the following claim.

    \begin{claim}[Candidate Block Witness]
        \label{clm:candidate-block-witness}
        Fix $i, j \in [N]$.
        If $K$ is a block witness, then $K$ is a candidate block witness with respect to $p$ for any choice of prime $p$.

        Let $p \sim [P, 2P]$ be a randomly sampled prime.
        If $K$ is not a block witness, then $K$ is not a candidate block witness with respect to $p$ with probability $\geq 1/2$ over the choice of $p$.
    \end{claim}

    \begin{proof}[Proof of \Cref{clm:candidate-block-witness}]
        The first claim is trivial, so we focus on the second.
        Since $K$ is not a block witness, for $k \in [K]$ define
        \begin{equation*}
            \Delta_{k} := A[i, k] + B[k, j] - C[i, j] \neq 0 \text{.}
        \end{equation*}
        Note that $0 < |\Delta_{k}| \leq 3 U$ so that $\Delta_{k}$ has at most $O(\log_P U) = \bigO{\frac{\log U}{\log P}}$ prime factor in the range $[P, 2 P]$.
        By the prime number theorem, there are $\Theta(P / \log P)$ primes in the range $[P, 2P]$ so that for a random choice of prime $p \sim [P, 2P]$, the probability that $\Delta_{k} \equiv 0 \mod p$ is at most 
        \begin{equation*}
            \bigO{\frac{\log U}{\log P} \frac{\log P}{P}} = \bigO{\frac{\log U}{P}} \text{.}
        \end{equation*}
        A union bound ensures that $\Delta_{k} \equiv 0 \mod p$ for any $k \in K$ with probability at most $O(B \log U / P)$.
        By our choice of $P \gg B \log U$, the probability can be upper bounded by $1/2$.
    \end{proof}

    This motivates the following approach: sample $O(\log n)$ primes independently in the range $[P, 2P]$. 
    We say $K$ is a \emph{candidate block witness} if $K$ is a candidate block witness with respect to all sampled primes $p$.
    Any block that is a block witness will be a candidate block witness, whereas any block that is not a block witness will be a candidate block witness with inverse polynomial probability.
    Our first goal is to generate candidate block witnesses algorithmically.
    To do so, we give a subroutine that given block $K$ and prime $p$, checks whether $K$ is a candidate block witness with respect to $p$.

    \begin{lemma}[Candidate Block Witness Verification]
        \label{lem:candidate-block-wit-ver}
        Let $K$ be a block and $p \in [P, 2P]$ be a prime.
        For every $\varepsilon < \varepsilon^* := (\ln 4) / (5 \ln 10) = 0.1204\dots$ and every $\kappa > 0$ there is a $\gamma > 0$ such that the following holds.
        
        If $2 B P \leq N^{\varepsilon}$, there is an algorithm that decides for every pair $(i, j)$ whether $K$ is a candidate block witness with respect to $p$ for $(i, j)$ in time 
        \begin{equation*}
            \bigtO{N^2 + NBP^2 + N^2 P^{1 - \gamma} B^{- \gamma} + (BP)^{\kappa - \gamma} N^2} \text{.}
        \end{equation*}
    \end{lemma}

    \begin{proof}[Proof of \Cref{lem:candidate-block-wit-ver}]
        Here, we will use the \Cref{thm:av-min-plus}, restated below.

        \ThinMatrixProduct* 
        
        Note that since $N^2 / D^{\kappa}$ is the output size, we must have $\gamma \leq \kappa$.
        We begin with the algorithm.

        \begin{mdframed}
            {\bf Verify Candidate Block Witness.}
            \begin{enumerate}
                \item For $r \in \Z_{p}$, let $W_{r} := \set{(i, j) \given C[i, j] \equiv r \mod p}$.
                \item For $r \in \Z_{p}$, construct boolean matrices $X, Y_{r}$ with entries
                \begin{equation*}
                    X[i, (k, a)] = \ind{A[i, k] \equiv a \mod p} \text{,}~~~~ Y_{r}[(k, a), j] = \ind{B[k, j] \equiv r - a \mod p} \text{.}
                \end{equation*}
                \Comment{($X$ can be constructed once as it does not depend on $r$.)}
                \item For $r \in Z_{p}$, compute $(X Y_{r})[(i, j)]$ for $(i, j) \in W_{r}$.
                \item For $r \in \Z_{p}$ and $(i, j) \in W_{r}$, return $\ind{(X Y_{r})[i, j] > 0}$.
            \end{enumerate}
        \end{mdframed}

        We begin with correctness.
        For any pair $(i, j) \in W_{r}$,
        \begin{align*}
            (XY_{r})[i, j] &= \sum_{(k, a)} \ind{A[i, k] \equiv a \mod p} \ind{B[k, j] \equiv r - a \mod p} \\
            &= \sum_{k} \sum_{A[i, k] \equiv a \mod p} \ind{A[i, k] + B[k, j] \equiv a + (r - a) \equiv r \equiv C[i, j] \mod p} \\
            &= \sum_{k}  \ind{A[i, k] + B[k, j] \equiv C[i, j] \mod p} \text{.}
        \end{align*}
        The first equality follows by definition of the matrix product.
        The second equality follows by $(i, j) \in W_{r}$.
        The third equality follows by the fact that $A[i, k] \equiv a \mod p$ for exactly one value of $a \in \Z_{p}$.
        Thus, $K$ is a candidate block witness with respect to $p$ for $(i, j) \in W_{r}$ if and only if $(XY_{r})[i, j] > 0$, establishing correctness.

        To bound the running time, fix $r$ and examine the shape of $X$ and $Y_{r}$.
        Note that $X$ has dimension $N \times B p$, $Y_{r}$ has dimension $B p \times N$ and both can be constructed in linear time.
        Thus, in order to apply \Cref{thm:av-min-plus}, we require $D := B p \leq 2 B P \leq N^{\varepsilon}$.
        Furthermore, if $W_{r} \geq N^2 / D^{\kappa}$, we can partition $W_{r}$ into $O(|W_{r}| D^{\kappa}/N^2)$ disjoint subsets of size at most $N^2 / D^{\kappa}$.
        Applying \Cref{thm:av-min-plus} to each disjoint subset, the total time to compute all $(XY_{r})[i, j]$ for all $(i, j) \in W_{r}$ over all $r$ is therefore
        \begin{align*}
            \sum_{r} \bigO{\frac{N^2 \log^2 D}{D^{\gamma}} \left( 1 + \frac{|W_{r}| D^{\kappa}}{N^2} \right)} = \bigtO{\frac{N^2 P}{D^{\gamma}} + D^{\kappa - \gamma} \sum_{r} |W_{r}|} = \bigtO{\frac{N^2 P}{D^{\gamma}} + D^{\kappa - \gamma} N^2} \text{.}
        \end{align*}
        Constructing $X$ requires $O(NBP)$ time, while constructing $Y_{r}$ over all $r$ requires $O(NBP^2)$ time.
        Finally, note that the remaining steps requires $O(N^2)$ time.
    \end{proof}

    \paragraph{Step 3: Verifying the Sparse Certificate}
    Given \Cref{lem:candidate-block-wit-ver}, we present our algorithm for $(\min, +)$-Witness Equality. 
    In this final step, we search all candidate block witnesses and check if they are true block witnesses.
    As discussed, we begin by sampling random primes and generating all candidate block witnesses for every pair $(i, j)$.
    
    \begin{mdframed}
        {\bf Generate Candidate Block Witnesses.}
        \begin{enumerate}
            \item Sample $\ceil{\log S}$ random primes independently in the range $[P, 2P]$.
            \item Using \Cref{lem:candidate-block-wit-ver}, verify for every pair $(i, j)$ and block $K$, whether $K$ is a candidate block witness with respect to all sampled primes. Return $K$ is a candidate block witness if and only if yes.
        \end{enumerate}
    \end{mdframed}

    Correctness of the above procedure follows by the correctness of \Cref{lem:candidate-block-wit-ver} and the definition of candidate block witnesses.
    The time required over all $O(N/B)$ blocks $K$ and $O(\log S)$ primes $p$ is
    \begin{equation*}
        \bigtO{\frac{N^3}{B} + N^2 P^2 + \frac{N^3 P^{1 - \gamma}}{B^{1 + \gamma}} + \frac{N^3 (BP)^{\kappa - \gamma}}{B}} \text{.}
    \end{equation*}
    For a pair $(i, j)$, let $K(i, j)$ denote the set of candidate block witnesses.
    In the second phase, we test whether a candidate block witness is a block witness.
    We use the following algorithm for every pair $(i, j)$.

    \begin{mdframed}
        {\bf Verify Block Witnesses.}
        \begin{enumerate}
            \item For each candidate block witness $K \in K(i, j)$ enumerate $k \in K$ and check $A[i, k] + B[k, j] = C[i, j]$.
            \item If any check succeeds, return $(i, j)$ is accurate and terminate.
            \item If no check succeeds, return $(i, j)$ is inaccurate.
        \end{enumerate}
    \end{mdframed}

    Correctness of the algorithm follows from the fact that \Cref{clm:candidate-block-witness} implies that every block witness is a candidate block witness.
    Thus, if $(i, j)$ is accurate, we will eventually iterate over a block witness in $K(i, j)$ and enumerate the witness $k$ such that $A[i, k] + B[k, j] = C[i, j]$.
    On the other hand, if $(i, j)$ is not accurate, no check will succeed.

    To bound the running time, we separately bound the time required to process block witnesses and false candidates (i.e. candidate block witnesses that are not block witnesses).
    Since we terminate after any successful check, we process at most $1$ block witness for every pair $(i, j)$, resulting in total time $O(N^2 B)$ to process all block witnesses.
    To bound the time required to process false candidates, we bound the total number of false candidates.
    From \Cref{clm:candidate-block-witness}, a block $K$ that is not a block witness is a candidate block witness with probability at most $2^{-\ceil{\log S}} \leq S^{-1}$.
    Since there are $N^2$ $(i, j)$ pairs and $O(N/B)$ blocks per pair, the total number of false candidates is in expectation at most $O(N^3 / (BS))$, resulting in $O(N^3 / S)$ expected time to process all false candidates.
    Overall, the expected runtime of our algorithm is 
    \begin{equation*}
        \bigtO{\frac{N^3}{B} + N^2 P^2 + \frac{N^3 P^{1 - \gamma}}{B^{1 + \gamma}} + \frac{N^3 (BP)^{\kappa - \gamma}}{B} +  N^2 B + N^3 / S} \text{.}
    \end{equation*}
    Recall that we require $P \gg B \log U$, so that minimizing $P$ subject to this constraint, the running time is
    \begin{equation*}
        \bigtO{\frac{N^3}{B} + N^2 B^2 \log^2 U + N^3 B^{-2 \gamma} \log U + N^3 B^{2 \kappa - 2 \gamma - 1} \log U + N^2 B + N^3 / S} \text{.}
    \end{equation*}
    Setting $S = N$ and simplifying (e.g. recalling $\kappa \geq \gamma$), we obtain the running time
    \begin{equation*}
        \bigtO{\frac{N^3 B^{2(\kappa - \gamma)}}{B} \log U + N^2 B^2 \log^2 U + N^3 B^{- 2 \gamma} \log U} \text{.}
    \end{equation*}
    From Table 2 of \cite{alman2026truly}, we instantiate the constants as $\varepsilon = 0.056$, $\kappa = 0.5$, and $\gamma = 0.0653$.
    Then, note that $2(\kappa - \gamma) - 1 = - 2 \gamma$ so that it suffices to balance 
    \begin{equation*}
        \bigtO{N^3 B^{-2 \gamma} \log U + N^2 B^2 \log^2 U} \text{.}
    \end{equation*}
    Subject to the constraint $2 B P \leq N^{\varepsilon}$ and $P \gg B \log U$, we set $B \ll N^{\varepsilon/2} / \sqrt{\log U}$ for some sufficiently small constant.
    Since \Cref{thm:av-min-plus} requires that $\varepsilon < \varepsilon^*$ and $\gamma \leq \kappa$ which we have chosen to be $\kappa = 1/2$, the running time is 
    \begin{equation*}
        \bigtO{N^{3 - \gamma \varepsilon} \log U} = \bigtO{N^{2.9963432} \log U} \text{.}
    \end{equation*}
\end{proof}

We prove the remaining required reduction to $(\min, +)$-Witness Equality for completeness.
The proof is essentially identical to \cite{jin2026deterministic}, with the only change being that the input and output matrices are not required to be monotone.

\begin{proof}[Proof of \Cref{lem:verification-reduction}]
    We recurse on the value of $U$.
    In the base case where $U = 1$, return $(A*B)[i, j] = 2$ for all $(i, j)$.

    Then, suppose $U > 1$.
    Without loss of generality, suppose $U$ is a power of $2$.
    We begin by constructing matrices $A', B'$ as follows
    \begin{equation*}
        A' = \ceil{A/2} \text{,}~~~~ B' = \ceil{B/2} 
    \end{equation*}
    and recursively compute $X'=A'* B'$. 
    It is easy to see that $A', B'$ have values in $[U/2]$.
    
    Let $X=A\star B$. 
    For every triple $i,k,j$, observe that
    \begin{equation*}
        2(A'[i, k] + B'[k, j] - 1) \leq A[i, k] + B[k, j] \leq 2(A'[i, k] + B'[k, j]) \text{.}
    \end{equation*}
    Taking minima over $k$ yields
    \begin{equation*}
        2 X'[i, j] - 2 \leq X[i, j] \leq 2 X'[i, j]
    \end{equation*}
    for every pair $(i, j)$.
    Thus, $X[i, j] \in \set{2 X'[i, j] - 2, 2 X'[i, j] - 1, 2 X'[i, j]}$.

    For each $s \in \set{-2, -1, 0}$ run the $(\min, +)$-Witness Equality algorithm with candidate matrix $C_{s}[i, j] = 2 X'[i, j] + s \in [2U]$ since $2 \leq X'[i, j] \leq U$.
    For every $(i,j)$, return the smallest $2 X'[i, j] + s$ for which the $(\min, +)$-Witness Equality algorithm reports that there exists $k$ such that
    \begin{equation*}
        A[i, k] + B[k, j] = 2 X'[i, j] + s \text{.}
    \end{equation*}
    Note that this algorithm is correct since $X[i, j] \in \set{2 X'[i, j] - 2, 2 X'[i, j] - 1, 2 X'[i, j]}$, so that the minimum value with a witness must be the true value of the $(\min, +)$ product.
    
    Let $M(N, U)$ denote the time required to compute the product of two $N \times N$ matrices with entries in $[U]$.
    By the master theorem, the total time required is
    \begin{equation*}
        M(N, U) = M(N, U/2) + \bigO{(T(N, U) + N^2)} = \bigO{(T(N, U) + N^2) \log U} \text{.}
    \end{equation*}
\end{proof}

\subsection{Derandomization}

Next, we discuss how to derandomize \Cref{thm:faster-min-plus}.

\begin{proof}[Proof of \Cref{thm:det-faster-min-plus}]
    Observe that in the proof of \Cref{thm:faster-min-plus} the only randomization required is to sample primes in Step 2, finding a sparse certificate. 
    We discuss how to choose these primes deterministically below. 
    To exploit the resulting sparsity, we search the candidate blocks depth-first and stop after finding the first witness for each output pair. 
    We use the data structure version of \Cref{thm:av-min-plus} (see Corollary 31 of \cite{alman2026truly}), with parameters depending on the level, to answer the queries made by this search.

    Fix parameters $1\leq B,P\leq N$ whose values we will determine later. 
    We may assume $U\geq2$ and that $B$ is a power of two. 
    We begin by splitting the inner dimension $[N]$ into $O(N/B)$ contiguous blocks of size at most $B$. 
    We call these blocks level-$0$ blocks, and denote one such block by $K^{(0)}$. 
    At level $0\leq\ell\leq\log B$, each level-$0$ block is split into at most $2^\ell$ contiguous blocks of size at most $B^{(\ell)}:=B/2^\ell$. 
    Note that $B=B^{(0)}$. 
    We let $K^{(\ell)}$ denote a level-$\ell$ block and let $h \leq \log B$ denote the last level that we search. As before, we say that $K^{(\ell)}$ is a block witness for $(i,j)$ if there exists $k\in K^{(\ell)}$ such that $A[i,k]+B[k,j]=C[i,j]$. 
    For each triple $(i,k,j)$, define $\Delta_{ikj}=A[i,k]+B[k,j]-C[i,j]$, so that we are looking, for each $(i,j)$, for some $k$ such that $\Delta_{ikj}=0$.

    Let us begin by recalling the properties required by the randomized algorithm.
    Let $\calP$ denote all primes between $[P, 2P]$.
    Consider some prime $p \in \calP$.
    We say a block $K^{(\ell)}$ is a candidate block witness for $(i, j)$ with respect to $p$ if there exist $k \in K^{(\ell)}$ such that $A[i, k] + B[k, j] \equiv C[i, j] \mod p$, and $K^{(\ell)}$ is a false candidate if it is a candidate block witness but not a block witness.
    Following the proof of \Cref{clm:candidate-block-witness}, we can show that for a randomly chosen prime $p \in \calP$, any block $K^{(\ell)}$ is a false candidate with probability at most $\bigO{\frac{B^{(\ell)} \log U}{P}}$.
    Thus, we may choose $P = \Theta(B \log U)$ to bound this probability by a constant, say $1/2$.
    In fact, by independently selecting $s$ random primes $p \in \calP$, one can reduce this probability to $2^{-s}$.
    To derandomize the above approach, we will instead show how to deterministically find primes $p \in \calP$ that reduce the number of false candidates.

    At level $\ell$, we will ensure that the number of false candidates is upper bounded by a function $\Phi^{(\ell)} := \bigTh{\frac{N^3}{B 2^{7 \ell}}}$, a bound which is trivially satisfied at $\ell = 0$.
    We will deterministically choose a constant number of primes such that the number of false level $\ell$ candidates decreases by a factor of $2^{7}$.
    We know that a random prime will eliminate all but a $\bigO{\frac{B^{(\ell)} \log U}{P}} = O(2^{-\ell})$ fraction of false candidates.
    Thus, a natural attempt is to deterministically choose a prime $p$ that eliminates the largest fraction of false candidates (or minimize $\Phi^{(\ell)}$).
    Unfortunately, both the quantity of false candidates eliminated by a given prime $p$  (and $\Phi^{(\ell)}$) may be hard to compute, as it may involve checking whether a candidate is genuine or not (see \cite{jin2026deterministic} for further discussion).
    Instead, we find a function $Z(p)$ that is a constant-factor approximation of $\Phi^{(\ell)}$ (\Cref{lem:z-is-phi-estimate}) and we can efficiently compute (\Cref{lem:compute-z}).
    This indirect approach of finding an estimate of the desired bound that is easier to compute is similar to \cite{jin2026deterministic}, although our estimator is distinct and tailored to our specific application.
    Thus, choosing a prime $p$ that minimizes $Z(p)$ suffices to ensure that the number of false candidates decreases by a factor of $O(2^{-\ell})$.
    Iteratively choosing a constant number of such primes at level $\ell$ (each time ensuring we minimize $Z(p)$) decreases $\Phi^{(\ell)}$ by the desired $2^{7 \ell}$ factor.

    Given the above hierarchy, we proceed to search blocks in a depth-first order.
    For any given pair $(i, j)$, iterate over all level $0$ blocks.
    At each block, we will check whether it is a candidate block (as before with \cite{alman2026truly}), and if so process its children in depth-first order.
    At the final level of the hierarchy, we enumerate elements in all candidate blocks, and terminate the search for $(i, j)$ whenever a true witness is found.
    The depth-first order of the traversal ensures that at level $\ell$, at most $1$ true block witness is examined.
    Thus, the total number of level $\ell$ blocks examined at level $\ell$ is at most $O(N^2 + \Phi^{(\ell - 1)})$, since we must also examine the children of false level $\ell - 1$ candidates.
    Since $\Phi^{(\ell)}$ is decreasing in $\ell$, our algorithm makes progressively fewer queries to the matrix product computed by \cite{alman2026truly} at higher levels $\ell$, but incurs higher preprocessing cost as the inner dimension $D^{(\ell)} := B^{(\ell)} P$ is decreasing in $\ell$.
    Thus, instead of using \Cref{thm:av-min-plus}, we use the data structure version of Theorem 25 (see Corollary 31 of \cite{alman2026truly}) which allows us to obtain faster preprocessing at the cost of slower queries.

    We elaborate on the high-level approach described above below.

    \paragraph{Step 1: Finding a Good Modulus.}
    Let $t\geq1$ be a constant whose value we will determine later. 
    At each level $\ell$, we pick a set of primes $\calP^{(\ell)}$ of size $t$. 
    We say that a block $K^{(\ell)}$ is a candidate block witness for $(i,j)$ with respect to $\calP^{(\ell)}$ if it is a candidate block witness with respect to every $p\in\mathcal P^{(\ell)}$. 
    A triple $(i,j,K^{(\ell)})$ is a level-$\ell$ candidate if, for every $0\leq r\leq \ell$, the ancestor $K^{(r)}\supseteq K^{(\ell)}$ is a candidate block witness for $(i,j)$ with respect to $\calP^{(r)}$. 
    A level-$\ell$ candidate is false if $K^{(\ell)}$ is not a block witness. 
    We have the following simple claim.

    \begin{claim}
        \label{clm:level-contain}
        If $K^{(\ell)}$ is a block witness for $(i,j)$, then $(i,j,K^{(\ell)})$ is a level-$\ell$ candidate. If $\ell\geq1$ and $(i,j,K^{(\ell)})$ is a level-$\ell$ candidate, then $(i,j,K^{(\ell-1)})$ is a level-$(\ell-1)$ candidate for the unique parent $K^{(\ell-1)}\supseteq K^{(\ell)}$.
    \end{claim}

    \begin{proof}[Proof of \Cref{clm:level-contain}]
        If $k\in K^{(\ell)}$ satisfies $\Delta_{ikj}=0$, then it belongs to every ancestor of $K^{(\ell)}$ and satisfies $\Delta_{ikj}\equiv0\pmod p$ for every prime $p$. 
        Thus every ancestor passes the tests for its level. The second assertion follows directly from the definition.
    \end{proof}

    For a given triple $(i, j, K^{(\ell)})$ and modulus $d$, let 
    \begin{equation*}
        F(i, j, K^{(\ell)}, d) = \left|\set{k \in K^{(\ell)} \mid \Delta_{ikj} \neq 0 \wedge \Delta_{ikj} \equiv 0 \mod d}\right| 
    \end{equation*}
    denote the number of false witnesses for $(i, j)$ in $K^{(\ell)}$ with respect to $d$.
    We define the following potential function
    \begin{equation*}
        \Phi^{(\ell)} = \sum_{i, j, K^{(\ell)}} \prod_{p \in \calP^{(\ell)}} F(i, j, K^{(\ell)}, p) \text{.}
    \end{equation*}
    We claim that this potential upper bounds the number of false level $\ell$ candidates.

    \begin{claim}
        \label{clm:potential-ub}
        The total number of false level $\ell$ candidates is at most $\Phi^{(\ell)}$.
    \end{claim}

    \begin{proof}[Proof of \Cref{clm:potential-ub}]
        Suppose $(i, j, K^{(\ell)})$ is a false level $\ell$ candidate.
        Then, for each $p \in \calP^{(\ell)}$, there exists $k \in K^{(\ell)}$ such that $\Delta_{ikj} \equiv 0 \mod p$ and $\Delta_{ikj} \neq 0$.
        In particular, we have $F(i, j, K^{(\ell)}, p) > 0$ for all $p \in \calP^{(\ell)}$, so that $\prod_{p} F(i, j, K^{(\ell)}, p) > 0$.
    \end{proof}

    Thus, it suffices to bound $\Phi^{(\ell)}$.
    Towards this goal, our goal at each level $\ell$ is to choose $\calP^{(\ell)}$ satisfying the following invariant
    \begin{equation}
        \label{eq:potential-bound}
        \Phi^{(\ell)} = \bigO{\frac{N^3}{B^{(\ell)}} \left(\frac{B^{(\ell)} \log U}{P} \right)^{t}} = \bigO{\frac{N^3 B^{t - 1} \log^{t} U}{2^{\ell(t - 1)} P^{t}}} \text{.}
    \end{equation}
    As in \Cref{thm:faster-min-plus}, we set $P \gg B \log U$ for some sufficiently large constant so that
    \begin{equation*}
        \Phi^{(\ell)} = \bigO{\frac{N^3 2^{\ell}}{B} 2^{- \ell t}} = \bigO{\frac{N^3}{B 2^{\ell(t - 1)}}} \text{.}
    \end{equation*}

    \begin{lemma}
        \label{lem:finding-good-primes}
        Let $\ell\geq0$ and let $t\geq1$ be a constant. 
        There exists a deterministic algorithm that finds a set $\calP^{(\ell)}$ consisting of $t$ distinct primes in $[P,2P]$ such that $\Phi^{(\ell)}$ satisfies \Cref{eq:potential-bound}, in time
        \begin{equation*}
            \bigtO{ P^{t+1} \log^{t + 1}(U) \TMUL \left(N,N(B^{(\ell)})^{t - 1},N\right)} \text{.}
        \end{equation*}
    \end{lemma}

    \begin{proof}[Proof of \Cref{lem:finding-good-primes}]
        Note that $|\Delta_{ikj}|\leq4U$, so every nonzero $\Delta_{ikj}$ has at most $H:=\lceil\log(4U)\rceil$ distinct prime factors. 
        Let $Q$ consist of the first $2tH$ primes. 
        We will ensure a choice of $P$ such that for sufficiently large $N$, these primes are $<P$. 
        We construct $\calQ$ deterministically in $\widetilde O(H)$ time. 
        
        We will construct $\calP^{(\ell)}$ iteratively, beginning with $\calP^{(\ell)}_{0} = \emptyset$.
        Throughout the construction, we choose the next prime from $\mathcal P\setminus\mathcal P_s^{(\ell)}$, so all chosen primes are distinct.
        For $0 \leq s \leq t$, define the potential
        \begin{equation*}
            \Phi^{(\ell)}_{s} := \sum_{i, j, K^{(\ell)}} \prod_{p \in \calP^{(\ell)}_{s}} F(i, j, K^{(\ell)}, p) 
        \end{equation*}
        so that $\Phi^{(\ell)} = \Phi^{(\ell)}_{t}$.
        We maintain the following invariant:
        \begin{equation}
            \label{eq:potential-loop-invariant}
            \Phi^{(\ell)}_{s} = \bigO{\frac{N^3}{B^{(\ell)}} \left(\frac{B^{(\ell)} \log U}{P} \right)^{s}}
        \end{equation}
        which is trivially satisfied at $s = 0$ since there are $O(N^3/B^{(\ell)})$ triples $(i, j, K^{(\ell)})$.
        Given $\calP^{(\ell)}_{s}$, we describe how to choose $p_{s + 1}$ to obtain $\calP^{(\ell)}_{s + 1}$.
        At a high level, our goal is to choose $p$ so that $\Phi^{(\ell)}_{s + 1}$ is minimized.
        However, since it is not clear how to compute $\Phi^{(\ell)}_{s}$ efficiently directly, we must take an alternative approach.
        
        Define
        \begin{align*}
            C(i, j, K^{(\ell)}, p) &= \left|\set{k \in K^{(\ell)} \mid \Delta_{ikj} \equiv 0 \mod p} \right| \\
            Z(i, j, K^{(\ell)}) &= \left|\set{k \in K^{(\ell)} \mid \Delta_{ikj} = 0} \right| 
        \end{align*}
        so that
        \begin{equation*}
            F(i, j, K^{(\ell)}, p) = C(i, j, K^{(\ell)}, p) - Z(i, j, K^{(\ell)}) \text{.}
        \end{equation*}
        For any fixed $p'\in\calP \setminus \calP_s^{(\ell)}$, let $\calP_{s+1}^{(\ell)}=\calP_s^{(\ell)}\cup\{p'\}$ and define
        \[
         Z_{s+1}(p'):=
         \sum_{q\in \calQ}\sum_{i,j,K^{(\ell)}}
         \prod_{p\in \calP_{s+1}^{(\ell)}}
         \left(C(i,j,K^{(\ell)},p)-C(i,j,K^{(\ell)},pq)\right).
        \]
        Here $\calP_{s+1}^{(\ell)}$, and hence $\Phi_{s+1}^{(\ell)}$, depends on $p'$. We first show that $Z_{s+1}(p')$ approximates $\Phi_{s+1}^{(\ell)}$, and then describe how to compute it.

        \begin{lemma}
            \label{lem:z-is-phi-estimate}
            Let $\calP^{(\ell)}_{s}$ be a set of $s$ distinct primes. 
            Then, for any prime $p' \in \calP \setminus \calP^{(\ell)}_{s}$, 
            \begin{equation*}
                \frac{|\calQ|}{2} \Phi^{(\ell)}_{s + 1} \leq Z_{s + 1}(p') \leq |\calQ| \Phi^{(\ell)}_{s + 1} \text{.}
            \end{equation*}
        \end{lemma}

        \begin{proof}[Proof of \Cref{lem:z-is-phi-estimate}]
            Fix a triple $(i, j, K^{(\ell)})$.
            We begin by examining $\Phi^{(\ell)}_{s + 1}$.
            The product $\prod_{p \in \calP^{(\ell)}_{s + 1}} F(i, j, K^{(\ell)}, p)$ counts tuples $(k_1, \dots, k_{s + 1}) \in K^{(\ell)}$ such that $\Delta_{ik_{u}j} \neq 0$ and $\Delta_{i k_{u} j} \equiv 0 \mod p_{u}$ for $1 \leq u \leq s + 1$.
            We say such a tuple $(i, j, K^{(\ell)}, k_1, \dots, k_{s + 1})$ is \emph{bad}, so that $\Phi^{(\ell)}_{s + 1}$ counts exactly the number of bad tuples.

            Now, let us examine $Z_{s + 1}(p')$.
            Observe
            \begin{align*}
                C(i, j, K^{(\ell)}, p) - C(i, j, K^{(\ell)}, p q) &= \left| \set{k \in K^{(\ell)} \mid \Delta_{ikj} \equiv 0 \mod p \wedge \Delta_{ikj} \not\equiv 0 \mod q} \right| 
            \end{align*}
            Since $p \neq q$, no triple $\Delta_{ikj} = 0$ contributes to the difference.
            Furthermore, $\prod_{p} \left( C(i, j, K^{(\ell)}, p) - C(i, j, K^{(\ell)}, p q) \right)$ counts tuples $(k_1, \dots, k_{s + 1})$ such $\Delta_{ik_uj} \not\equiv 0 \mod q$ but $\Delta_{ik_uj} \equiv 0 \mod p_{u}$ for $u \in [s + 1]$.
            We say a tuple $(i, j, K^{(\ell)}, k_1, \dots, k_{s + 1})$ is counted by $q$ if it satisfies the above condition.

            By definition, for any $q \in \calQ$, every tuple that is counted by $q$ is bad, yielding the upper bound.
            To argue the lower bound, it suffices to show that any bad tuple is counted by at least $1/2$ fraction of $\calQ$.
            To see this, observe that each $|\Delta_{i k_u j}| \leq 4 U$ and therefore has at most $H = O(\log U)$ distinct prime factors.
            Combined, there are at most $t H$ distinct primes that factor at least one of of $\Delta_{i k_1 j}, \dots, \Delta_{i k_{s + 1} j}$.
            Since $|\calQ| \geq 2 t H$, there are at least $t H$ primes $q \in \calQ$ that do not divide any of $\Delta_{i k_{u} j}$, so that any bad tuple is counted by at least $t H$ of the primes in $\calQ$.
        \end{proof}

        Next, we describe how to estimate $Z_{s + 1}(p')$.

        \begin{lemma}
            \label{lem:compute-z}
            For every constant $s\geq0$, there is a deterministic algorithm that computes $Z_{s+1}(p')$ for all $p'\in\mathcal P\setminus\mathcal P_s^{(\ell)}$ in time
            \[
             \tO{P^{s+2} \log^{s + 2} U}
            \TMUL \left(N,N(B^{(\ell)})^s,N\right) \text{.}
            \]
        \end{lemma}

        \begin{proof}[Proof of \Cref{lem:compute-z}]
            We compute $Z_{s+1}(p')$ separately for each $p'\in\mathcal P\setminus\mathcal P_s^{(\ell)}$. 
            Fix such a $p'$, write $\mathcal P_s^{(\ell)}=\{p_1,\ldots,p_s\}$, and let $p_{s+1}=p'$.
            
            Fix $q \in \calQ$ and a triple $(i, j, K^{(\ell)})$.
            For a set $S \subseteq [s + 1]$, we let $r_u = p_u$ if $u \not\in S$ and $p_u q$ otherwise.
            We may rewrite the product 
            \begin{equation*}
                \prod_{p} \left( C(i, j, K^{(\ell)}, p) - C(i, j, K^{(\ell)}, p q) \right) = \sum_{S \subseteq [s + 1]} (-1)^{|S|} \prod_{u = 1}^{s + 1} C(i, j, K^{(\ell)}, r_u) \text{.}
            \end{equation*} 
            For a fixed set $S$, we will compute
            \begin{equation}
                \label{eq:tuple-sum}
                (-1)^{|S|} \sum_{i, j, K^{(\ell)}} \prod_{u = 1}^{s + 1} C(i, j, K^{(\ell)}, r_u) \text{.}
            \end{equation}

            To do so, we construct matrices over the ring $\Z[z_1, \dots, z_{s + 1}]/(z_{1}^{r_{1}} - 1, \dots, z_{s + 1}^{r_{s + 1}} - 1)$ as follows
            \begin{align*}
                X[i, (K^{(\ell)}, k_1, \dots, k_{s + 1})] &= \prod_{u = 1}^{s + 1} z_{u}^{A[i, k_{u}] \mod r_{u}} \\
                Y[(K^{(\ell)}, k_1, \dots, k_{s + 1}), j] &= \prod_{u = 1}^{s + 1} z_{u}^{B[k_{u}, j] \mod r_{u}} \\
            \end{align*}
            The outer dimension of the matrices is $[N]$ while the inner dimension is
            \begin{equation*}
                \bigO{\frac{N}{B^{(\ell)}} \cdot (B^{(\ell)})^{s + 1}} = \bigO{\frac{N B^{s}}{2^{s \ell}}} \text{.}
            \end{equation*}
            The polynomials in the ring have degree less than $r_{u}$ in $z_{u}$ so that the total number of monomials is at most $R := \prod_{u=1}^{s+1} r_{u}$.
            Since every $r_u$ is either $p_u$ or $p_u q$, where $p_u\in[P,2P]$ and $q$ is $\tO{\log U}$, we have $R=\bigtO{P^{s + 1} \log^{s + 1}(U)}$. 
            For constant $s$, ring multiplication can be performed in $\tO{R}$ time by multivariate polynomial multiplication followed by reduction modulo $z_u^{r_u}-1$. 
            Thus the matrix product can be computed in time $\bigtO{P^{s + 1} \log^{s + 1}(U)} \TMUL \left(N,N(B^{(\ell)})^s,N \right)$.
            The computed entries are of the form
            \begin{equation*}
                (XY)[i, j] = \sum_{K^{(\ell)}, (k_1, \dots, k_{s + 1})} \prod_{u = 1}^{s + 1} z_{u}^{A[i, k_u] + B[k_u, j] \mod r_{u}} \text{.}
            \end{equation*}
            In particular, the coefficient of $\prod_{u = 1}^{s + 1} z_{u}^{C[i, j] \mod r_{u}}$ is exactly the number of tuples $(k_1, \dots, k_{s + 1})$ such that $\Delta_{i k_u j} \equiv 0 \mod r_{u}$ for all $1 \leq u \leq s + 1$.
            In particular, summing the coefficients of $\prod_{u = 1}^{s + 1} z_{u}^{C[i, j] \mod r_{u}}$ multiplied by $(-1)^{|S|}$ over all entries $(XY)[i, j]$, we obtain exactly the desired sum in \Cref{eq:tuple-sum}.
            Iterating over $O(P)$ candidate primes $p'$, all $q\in Q$, and all $2^{s+1}$ choices of $S$, we compute every $Z_{s+1}(p')$ in time
            \[
            \tO{P^{s+2} \log^{s + 2} U}
            \TMUL \left(N,N(B^{(\ell)})^s,N\right) \text{,}
            \]
            as desired.
        \end{proof}

        The above discussion motivates the following algorithm. Given
        $\calP_s^{(\ell)}$, compute $Z_{s+1}(p')$ for every
        $p'\in\mathcal P\setminus\mathcal P_s^{(\ell)}$, choose
        \[
         p^*\in\operatorname*{arg\,min}_{
         p'\in\mathcal P\setminus\mathcal P_s^{(\ell)}}Z_{s+1}(p'),
        \]
        and set $\mathcal P_{s+1}^{(\ell)}=
        \mathcal P_s^{(\ell)}\cup\{p^*\}$. 
        By \Cref{lem:compute-z}, summing over $0\leq s<t$ gives the time bound in \Cref{lem:finding-good-primes}.
        It remains to prove that the chosen primes satisfy \Cref{eq:potential-bound}.
        
        \begin{claim}
            \label{clm:greedy-valid}
            The above choice of $p^*$ maintains \Cref{eq:potential-loop-invariant}.
        \end{claim}
        \begin{proof}
            Fix $\mathcal P_s^{(\ell)}$. 
            For a uniformly chosen prime $p'\in\mathcal P\setminus\mathcal P_s^{(\ell)}$, \Cref{clm:candidate-block-witness} gives
            \begin{equation*}
                \Exp \left[ \Phi_{s+1}^{(\ell)}(p') \right] = \bigO{\frac{B^{(\ell)} \log U}{P}} \Phi^{(\ell)}_{s} \text{.}
            \end{equation*}
            Indeed, each tuple counted by $\Phi_s^{(\ell)}$ can be extended in at most
            $B^{(\ell)}$ ways, and a fixed nonzero difference is divisible by $p'$ with
            probability $O(\log U/P)$. 
            Excluding the constant number of already chosen primes changes this bound by only a constant factor.
            
            By \Cref{lem:z-is-phi-estimate} and the definition of $p^*$,
            \begin{equation*}
                \Phi_{s+1}^{(\ell)}(p^*) \leq \frac{2Z_{s+1}(p^*)}{|Q|} \leq \min_{p'}2 \Phi_{s+1}^{(\ell)}(p') \leq 2 \Exp_{p'} \left[ \Phi_{s+1}^{(\ell)}(p') \right] \text{.}
            \end{equation*}
            Starting from $\Phi_0^{(\ell)}=O(N^3/B^{(\ell)})$ and applying this bound
            at most constant $t$ times proves \Cref{eq:potential-loop-invariant}, and hence \Cref{eq:potential-bound}. 
        \end{proof}

        This concludes the proof of \Cref{lem:finding-good-primes}.
    \end{proof}

    \paragraph{Step 2: Finding and Verifying a Sparse Certificate.}
    \Cref{lem:finding-good-primes} allows us to find a set of good primes.
    We now describe how to use the set of good primes to compute a small sparse certificate (i.e. find the set of level $\ell$ candidates).
    In contrast to the randomized algorithm, we will interleave the steps of finding and verifying the sparse certificate.
    Concretely, we will use the data structure of \cite[Corollary~31]{alman2026truly} to test candidate blocks during a depth-first search. 

    \begin{theorem}[Corollary 31 of \cite{alman2026truly}]
        \label{thm:av-query}
        Let $c > 10$ and $0 < \theta < 0.9$, let $\rho_c := 9/(c - 1)$, and let
        \begin{equation*}
            g := \frac{\theta \ln(1/\rho_c)}{\ln 4} > 0 \text{ and } q := \frac{H(\theta) + \theta \ln 9}{\ln 4} \text{.}
        \end{equation*}
        Given as input matrices $X \in \Z^{N \times D}$ and $Y \in \Z^{D \times N}$ whose entries are integers of absolute value at most $N^{O(1)}$, where $2 \leq D \leq N^{\varepsilon}$ and $\varepsilon < R_{c}(g)$ where
        \begin{equation*}
            R_c(g) := \frac{\ln 4}{c \ln 10 - \frac{c}{2} H(1/c) - (c - 1) \ln 3 + g \ln 4} \text{,}
        \end{equation*}
        we can preprocess them deterministically in $O(N^2 \log^2 D/D^{g})$ time and space, after which any single entry $(XY)[i, j]$ in can be computed deterministically in $O(D^{q} \log D)$ time. 
        The constants hidden in the $O()$ depend on $c, \theta, \varepsilon$.
        Moreover, $R_c(0)$ increases to $\varepsilon^* = \ln 4/(5 \ln 10) = 0.1204\dots$ as $c$ decreases to $10$.
    \end{theorem}

    \Cref{thm:av-query} allows us to trade-off slower queries for faster preprocessing.
    At level $\ell$, we let $g_{\ell} = \frac{\theta_{\ell} \ln(1/\rho_c)}{\ln 4}$ so that $\theta_{\ell} = \frac{g_{\ell} \ln 4}{\ln(1/\rho_c)}$.
    
    \begin{lemma}
        \label{lem:level-data-structure}
        Fix a level $\ell$ and the set $\mathcal P^{(\ell)}$ of $t$ primes.
        Suppose $0 < \theta_{\ell} := \frac{g_{\ell} \ln 4}{\ln(1/\rho_c)} < 0.9$, $2B^{(\ell)}P\leq N^{\varepsilon}$, and $\varepsilon < R_c(g_{\ell})$.
        Given $g_{\ell}$, let $q_{c}(g_{\ell}) := \frac{H(\theta_{\ell}) + \theta_{\ell} \ln 9}{\ln 4}$.
        
        There is a deterministic data structure that returns $C(i,j,K^{(\ell)},p)$ for any block $K^{(\ell)}$, output pair $(i,j)$, and $p\in\mathcal P^{(\ell)}$, with preprocessing time $\bigtO{N^2P^2+\frac{N^3P}{B^{(\ell)}(B^{(\ell)}P)^{g_\ell}}}$ and query time $\widetilde O((B^{(\ell)}P)^{q_c(g_{\ell})})$.
    \end{lemma}
    
    \begin{proof}[Proof of \Cref{lem:level-data-structure}]
    For each block $K^{(\ell)}$, prime $p\in\mathcal P^{(\ell)}$, and residue
    $r \in\mathbb Z_p$, construct the matrices from \Cref{lem:candidate-block-wit-ver}
    \[
     X[i,(k,a)]=\mathbf 1[A[i,k]\equiv a\pmod p] \text{,}\qquad
     Y_r[(k,a),j]=\mathbf 1[B[k,j]\equiv r-a\pmod p].
    \]
    The inner dimension is $D=B^{(\ell)}p$; smaller blocks are padded with
    zero coordinates in both matrices.
    For a query $(i,j,K^{(\ell)},p)$, set $r=C[i,j]\bmod p$. 
    As in \Cref{lem:candidate-block-wit-ver} we have
    \[
     (XY_r)[i,j]=C(i,j,K^{(\ell)},p) \text{.}
    \]
    Apply \Cref{thm:av-query} to preprocess each input matrix pair.
    There are $O(NP/B^{(\ell)})$ products, since $t$ is constant. 
    Constructing all input matrices takes $O(N^2P^2)$ time, and their preprocessing and query costs give the claimed bounds.
    \end{proof}
    
    We now use these data structures to search the candidate blocks in depth-first order.
    For each output pair $(i,j)$, process all level-$0$ blocks in order. 
    At each visited block $K^{(\ell)}$, use the following algorithm.
    \begin{mdframed}
        \begin{enumerate}
            \item Query $C(i,j,K^{(\ell)},p)$ for each $p\in\mathcal P^{(\ell)}$.
            If any $C(i, j, K^{(\ell)}, p) = 0$, discard and proceed to the next block.
            \item If $\ell = h$, test $\Delta_{ikj} = 0$ for all $k \in K^{(h)}$ until a witness is found. If $\Delta_{ikj} = 0$ is found, return $(i, j)$ is accurate and proceed to the next $(i, j)$ output pair.
            If no $\Delta_{ikj} = 0$ is found, discard and proceed to the next block.
            \item If $\ell < h$, recurse on the first child $(i, j, K^{(\ell + 1)})$ and then recurse on the second child (in depth-first search order).
        \end{enumerate}
    \end{mdframed}
    If the full search finishes without finding a witness $\Delta_{ikj} = 0$, return that $(i,j)$ is inaccurate.
    
    \begin{lemma}
        \label{lem:det-search}
        Let $P=\Theta(B\log U)$, where the implied constant is sufficiently large, and let $0\leq h\leq \log B$. 
        Suppose the parameters in \Cref{lem:level-data-structure} are valid for every $0\leq \ell\leq h$. 
        Then the above algorithm solves $(\min,+)$-Witness Equality deterministically in time
        \begin{align*}
            \tilde{O}_{\log^{t + 1}(U)} \left(N^\omega B^{t-1} P^{t+1} + N^2 P^2 + \frac{N^2 B}{2^{h}} + \frac{N^3}{2^{ht}} + \sum_{\ell = 0}^{h} \left( \frac{N^3 P}{B^{(\ell)} (B^{(\ell)} P)^{g_{\ell}}} + \left( N^2 + \frac{N^3}{B 2^{\ell(t - 1)}} \right) (B^{(\ell)} P)^{q_{\ell}} \right) \right)
        \end{align*}
    \end{lemma}
    
    \begin{proof}[Proof of \Cref{lem:det-search}]
        Correctness follows from \Cref{clm:level-contain} as we return that $(i,j)$ is accurate only after finding an exact witness, and a search will not terminate until either an exact witness is found or all candidate blocks are explored.
    
        To bound the number of queries, fix $(i,j)$. At each level, the depth-first search visits at most one block witness. Indeed, once it enters a block containing an exact witness, it finds such a witness before returning from that block, and then stops. 
        Over all output pairs, at most $N^2$ visited blocks at any level are block witnesses. 
        The number of false candidates at level $\ell$ is at most $\Phi^{(\ell)}$ by \Cref{clm:potential-ub}.
        
        There are $O(N^3/B)$ tests at level $0$.
        At level $\ell\geq1$, every tested block is a child of a visited level-$(\ell-1)$ candidate. 
        Thus the number of queries at level $\ell$, over the constant number $t$ of primes in $\calP^{(\ell)}$, is
        \[
         O(N^2+\Phi^{(\ell-1)})
         =O\!\left(N^2+\frac{N^3}{B2^{\ell(t-1)}}\right).
        \]
        The same bound holds at level $0$. 
        Multiplying by the query time in \Cref{lem:level-data-structure} gives the query term in the statement.
        The final scans at level $h$ cost
        \[
         O\!\left(N^2B^{(h)}+\Phi^{(h)}B^{(h)}\right)
         =O\!\left(\frac{N^2B}{2^h}+\frac{N^3}{2^{ht}}\right).
        \]
    
        It remains to include preprocessing and \Cref{lem:finding-good-primes}.
        Summing \Cref{lem:level-data-structure} over the $O(\log N)$ levels gives the preprocessing terms. 
        Furthermore, we have
        \[
         \TMUL\left(N,N(B^{(\ell)})^{t-1},N\right)
         =\bigO{N^\omega(B^{(\ell)})^{t-1}} \text{.}
        \]
        \Cref{lem:finding-good-primes}, summed over all levels, therefore contributes
        $\tO{N^\omega B^{t-1} P^{t+1} \log^{t + 1} U}$.
        Combining the bounds proves the lemma.
    \end{proof}

    \paragraph{Putting it Together: Choosing the Parameters.}
    We instantiate our algorithm with parameters
    \[
     c=21,\qquad \varepsilon=0.056,\qquad
     \gamma=0.0653,\qquad t=8.
    \]
    where $c, \varepsilon, \gamma$ are chosen to match \Cref{thm:faster-min-plus} and $t$ is a constant whose choice we justify later.
    Choose
    \[
     B=\Theta\!\left(\frac{N^{\varepsilon/2}}{\sqrt{\log U}}\right),
     \qquad P=\Theta(B\log U),\qquad
     h=\left\lceil\frac{\varepsilon\gamma}{t}\log N\right\rceil.
    \]
    Round $B$ down to a power of two. 
    Choose the constant in $P$ sufficiently
    large, and then the constant in $B$ sufficiently small that $2BP\leq N^\varepsilon$. 
    Thus $BP=\Theta(N^\varepsilon)$. 

    Next, we aim to choose $g_{\ell}$ so that the parameters of \Cref{lem:level-data-structure} are satisfied and we may apply \Cref{lem:det-search}.
    Before choosing $g_{\ell}$ at level $\ell$, we give a quick discussion on how this choice might affect the performance of \Cref{lem:level-data-structure}.
    The constants in \Cref{thm:av-query} may depend on parameters $c, \varepsilon, \theta_{\ell}$ and therefore $g_{\ell}$.
    To make these constants uniform over all levels and input sizes, we round each chosen $g_\ell$ upward to the nearest multiple of $10^{-6}$ (or any arbitrarily small fixed constant).
    Since $g_{\ell}$ lies in a bounded range, there are $O(1)$ total choices of $g_{\ell}$ and the $O()$ term of \Cref{thm:av-query} can indeed be bounded by an absolute constant.
    
    Following the above discussion, at level $\ell$, choose $g_\ell$ by rounding
    \[
     \frac{\varepsilon\gamma\log N+\ell}
          {\varepsilon\log N-\ell}
    \]
    upward to the nearest multiple of $10^{-6}$, and set $q_\ell=q_{c}(g_\ell)$. 
    Note that our choice of $\gamma, t$ ensure $h < \varepsilon \log N$.
    Intuitively, since we make fewer queries at higher levels $\ell$, we are willing to trade off slower queries in exchange for faster preprocessing.
    
    We first check the hypotheses of \Cref{lem:level-data-structure}.
    Note $h \leq \log B$ by our choice of $t, \gamma$.
    For every $0\leq \ell \leq h$, we compute
    \begin{align*}
        \gamma &\leq g_{\ell} \leq g_{h} \leq \frac{\varepsilon \gamma + \frac{h}{\log N}}{\varepsilon - \frac{h}{\log N}} + 10^{-6} \leq \frac{\varepsilon \gamma (t + 1) + O(\frac{1}{\log N})}{\varepsilon(t - \gamma) + O(\frac{1}{\log N})} + 10^{-6} \leq \frac{\gamma(t + 1)}{t - \gamma} + 10^{-6} + \bigO{\frac{1}{\log N}} < 0.075
    \end{align*}
    for sufficiently large $N$.
    Therefore, since $c = 21$, we have
    \[
     0<\theta_\ell=\frac{g_\ell\ln4}{\ln(20/9)}<0.131<0.9 \text{.}
    \]
    Next, observe that $R_c$ is decreasing in $g$, so that
    \begin{equation*}
        \varepsilon=0.056 < R_{c}(0.075)\leq R_{c}(g_\ell),
    \end{equation*}
    where $R_{21}(0.075)=0.05663927 \dots$.
    We have already ensured $2B^{(\ell)}P\leq 2BP \leq N^\varepsilon$ so that \Cref{lem:level-data-structure} applies at every level. 

    By our choice of $B, P$ and $\varepsilon$, we can bound the time required for \Cref{lem:finding-good-primes} as $\bigtO{N^{\omega + t \varepsilon}} = \bigtO{N^{3 - \varepsilon \gamma}}$ as long as $\omega + t \varepsilon \leq 3 - \varepsilon \gamma$ or $t \leq \frac{3 - \omega - \varepsilon \gamma}{\varepsilon}$, which our choices satisfy.
    Note that our choice of $P$ ensures that every prime in $\calQ$ is smaller than $P$ for sufficiently large $N$.
    Note that this also bounds the time required by the $\tO{N^2 P^2}$ preprocessing term.
    
    We next bound the preprocessing time required by \Cref{lem:level-data-structure}. 
    Since $B^{(\ell)}P=\Theta(N^\varepsilon/2^\ell)$, we have
    \[
     (B^{(\ell)}P)^{g_\ell} = \bigTh{ \left( N^{\varepsilon}/2^{\ell} \right)^{g_{\ell}}}
     = \bigTh{2^{(\varepsilon \log N - \ell)g_{\ell}}} = \bigOm{N^{\varepsilon \gamma} 2^{\ell}} \text{.}
    \]
    Using $P/B^{(\ell)}=\Theta(2^\ell\log U)$, the preprocessing term at
    each level is consequently
    \[
     \bigtO{\frac{N^3P}{B^{(\ell)}(B^{(\ell)}P)^{g_\ell}}}
     = \bigtO{N^{3-\varepsilon\gamma}} \text{.}
    \]
    
    We now bound the query time, recall that $q_{21}$ is concave in $g_{\ell}$, since entropy is concave and $\theta=g\ln4/\ln(20/9)$ is linear in $g$. 
    Our choice of $\gamma$ satisfies
    \begin{align*}
        q_{21}(\gamma)+10^{-6}q_{21}'(\gamma)+\gamma &=0.49998460\ldots<\tfrac12,\\
        (1+\gamma)q_{21}'(\gamma)-q_{21}(\gamma)&=5.24066398\ldots<6.
    \end{align*}
    By concavity and our choice of $g_{\ell}$, we have
    \[
     q_{21}(g_{\ell}) \leq q_{21}(\gamma) + q_{21}'(\gamma)
     (g_{\ell} - \gamma) \leq q_{21}(\gamma) + q_{21}'(\gamma)
     \left(\frac{(1+\gamma)\ell}{\varepsilon\log N-\ell} + 10^{-6}\right) \text{.}
    \]
    Multiplying by $(\varepsilon\log N-\ell) > 0$ and using the two inequalities
    above gives
    \begin{align*}
        (\varepsilon\log N-\ell)q_{21}(g_{\ell}) &\leq (\varepsilon \log N - \ell) q_{21}(\gamma) + (1 + \gamma) \ell q_{21}'(\gamma) + 10^{-6} q'_{21}(\gamma) (\varepsilon \log N - \ell) \\
        &= ( q_{21}(\gamma) + 10^{-6} q_{21}'(\gamma)) \varepsilon \log N + ((1 + \gamma) q_{21}'(\gamma) - q_{21}(\gamma) - 10^{-6} q_{21}'(\gamma)) \ell
    \end{align*}
    Since $q_{21}$ is increasing, we have $q'_{21}(\gamma) > 0$ and we may drop the last term in the upper bound.
    Applying the two inequalities above, we have
    \begin{align*}
        (\varepsilon\log N-\ell)q_{21}(g_{\ell})  < \varepsilon \log N(0.5 - \gamma) + 6 \ell
    \end{align*}
    so that the time to compute one level $\ell$ query is $(B^{(\ell)} P)^{q_{21}(g_{\ell})} = \bigO{N^{\varepsilon(0.5 - \gamma)} 2^{6 \ell}}$.
    Note that we also have $q_{21}(g_{\ell}) \leq q_{21}(0.075) < 0.5$ since $q_{21}$ is increasing, so we may also upper bound $(B^{(\ell)} P)^{q_{21}(g_{\ell})} < (BP)^{1/2} = \tO{N^{\varepsilon/2}}$.
    From \Cref{lem:det-search}, the total time to compute all queries is
    \begin{align*}
        \bigtO{N^{2 + \varepsilon / 2} + \sum_{\ell} \frac{N^3}{B 2^{\ell(t - 1)}} N^{\varepsilon(0.5 - \gamma)} 2^{6 \ell}} &= \bigtO{N^{2 + \varepsilon/2} + \sum_{\ell} \frac{N^{3 - \varepsilon \gamma}}{2^{\ell(t - 7)}}} \\
        &= \bigtO{N^{2 + \varepsilon/2} + N^{3 - \varepsilon \gamma}}
    \end{align*}
    by our choice of $t > 7$.
    The final scan at level $h$ requires time $\bigO{\frac{N^2B}{2^h}+\frac{N^3}{2^{ht}}} = \bigtO{N^{2+\varepsilon/2}+N^{3-\varepsilon\gamma}} = \bigtO{N^{3 - \varepsilon \gamma}}$ as long as $2 + \varepsilon / 2 \leq 3 - \varepsilon\gamma$ or $\gamma \leq \frac{2 - \varepsilon}{2\varepsilon}$, which holds for our choices of $\varepsilon, \gamma$.
    We conclude the proof by observing that \Cref{lem:verification-reduction} costs a logarithmic factor.
\end{proof}

\bibliographystyle{alpha}
\bibliography{biblio}

\end{document}